\documentclass[fleqn,10pt]{IEEEtran}
\pdfoutput=1
\usepackage{graphicx}
\usepackage{amssymb}
\usepackage{amsmath}
\usepackage{amsthm}
\usepackage{epstopdf}
\usepackage{enumerate}
\usepackage{algorithm}
\usepackage[noend]{algpseudocode}

\DeclareMathOperator{\diag}{diag}

\DeclareMathOperator{\soft}{soft}
\DeclareMathOperator{\var}{var}

\newcommand{\thh}{\ensuremath{^\text{th}}}

\newcommand{\z}[1]{\mathcal{#1}}
\newcommand{\bo}[1]{\mathbf{#1}}

\theoremstyle{definition}
\newtheorem{defn}{Definition}
\newtheorem{assm}{Assumption}

\newtheorem{remark}{Remark}

\newtheorem{prop}{Proposition}

\newtheorem{experiment}{Experiment}

\title{Multichannel Linear Prediction for Blind Reverberant Audio Source Separation}
\author{\.Ilker Bayram and Sava\c{s}kan Bulek \vspace*{-0.5cm}
\\  \thanks{\.{I}. Bayram is with the Dept. of Electronics and Communications Eng., Istanbul Technical University, Istanbul, Turkey.  E-mail : ibayram@itu.edu.tr. 
S. Bulek is with Qualcomm Atheros, Inc., Auburn Hills, MI, USA. E-mail : sbulek@gmail.com.
}
}
\date{}

\begin{document}

\maketitle

\begin{abstract}
A class of methods based on multichannel linear prediction (MCLP) can achieve effective blind \emph{dereverberation} of a source, when the source is observed with a microphone array. We propose an inventive use of MCLP as a pre-processing step for blind \emph{source separation} with a microphone array. We show theoretically that, under certain assumptions, such  pre-processing reduces the original blind reverberant source separation problem to a non-reverberant one, which in turn can be effectively tackled using existing methods. We demonstrate our claims using real recordings obtained with an eight-microphone circular array in reverberant environments.
\end{abstract}

\section{Introduction} \label{sec:intro}
In anechoic environments, audio source separation with a microphone array can be achieved by  beamforming and post-filtering \cite{GannotBeam}. On the other hand, in reverberant environments,  reverberation causes significant degradation in the quality and intelligibility of the beamformer output, especially if the room impulse responses (RIR) are not taken into account or are not known precisely. This is a serious limitation because in practice, RIRs are usually not known with precision and they can vary wildly with respect to the positions of the microphone/source pair.

A related problem addressing reverberation without explicit RIR knowledge is called blind dereverberation \cite{naylor}. 
A successful class of methods for blind dereverberation fall under the framework of multi-channel linear prediction (MCLP) \cite{nak10p717,juk16, juk15p509}. These methods require a  microphone array and employ linear prediction to extract information about the properties of reverberation.
In this paper, we propose to employ MCLP, which is primarily used for single source dereverberation, as a pre-processing step in source separation. We show theoretically and demonstrate numerically that such a scheme not only dereverberates the sources but also helps in the subsequent separation.

In order to make this statement more precise, let us briefly discuss an MCLP formulation. Consider an array of microphones, recording a source in a reverberant room. Recent MCLP formulations \cite{nak10p717,juk16,juk15p509} select one of the microphones as the reference and use delayed versions of the whole set of observations to linearly predict the reference signal. The residual of prediction, which may also be interpreted as the `innovations' component, forms the dereverberated estimate of the source. In our scenario, we have multiple sources recorded with a microphone array and we would like to separate and dereverberate the sources. For this, suppose we apply MCLP for each possible choice of the reference microphone.  If we have $N$  microphones in the array, this results in $N$ hopefully dereverberant but mixed observations. We say `hopefully dereverberant' because the use of MCLP has previously been considered and justified solely for the case of a single source \cite{juk16}. Further, even if the MCLP outputs thus obtained are dereverberated, in order to achieve separation, we have to know the relation between the recordings. In this paper, we show that, under a mild assumption on the RIRs, MCLP actually suppresses reflections other than the direct source and preserves the phase of the direct source. Therefore, the MCLP preprocessing step essentially reduces the reverberant separation problem to an anechoic one.
This greatly simplifies the separation step because we can now use the known array geometry to estimate the direction of arrival and perform separation using this estimate. 
We argue that this effect is not dependent on the following DOA estimation and separation algorithms. In that sense, the proposed MCLP stage complements currently available source separation algorithms, improving their performance. Consequently, we do not have any reservations about which method to use following the MCLP stage.  Nevertheless, for the sake of completeness, we include a brief description of the geometric source separation (GSS) algorithm \cite{par02p352} followed by a post-filter.

\subsection*{Related Work}
Literature on dereverberation/source separation with a microphone array is vast. Here, we focus on blind \footnote{We use the term `blind' to refer to methods that do not require to know the RIRs. However, the blind methods we consider may use the known array geometry, which is independent of the characteristics of the environment.} methods for dereverberation and source separation for arrays with a known geometry.

For source separation, provided the directions of arrival (DOA) of the sources are known, and reverberation is limited, beamformers like the minimum variance distortionless response (MVDR) beamformer, along with a post-filter \cite{bitzer_chp,GannotBeam,simmer_chp}  can be effective. In MVDR, one aims to minimize the energy of the reconstructed signals, while preserving the signals from the DOAs. However, reverberation causes a source to contribute to different directions, reducing the effectiveness of such approaches. Geometric source separation (GSS)  \cite{par02p352} addresses this problem by also enforcing  statistical independence of the reconstructed sources.
This is achieved by selecting the beamformer weights so as to minimize a cost function that consists of the sum of a term imposing geometric constraints based solely on the array geometry and a term that promotes the statistical independence of the sources (motivated by blind source separation techniques \cite{ICA}). In practice, GSS outputs are still reverberant for moderate to highly reverberant rooms. GSS has been complemented with a post-filter in \cite{val04ICASSP,valin04IROS} to devise a high performance source-separation method that can operate in real time.
An alternative approach, based on the assumption that the activity of the different sources do not overlap in the time-frequency domain \cite{yil04p830}, has been proposed in   \cite{sou13p913}. However, the model in \cite{sou13p913} assumes that the room impulse responses are relatively short, compared to the analysis window size used in the STFT.

Despite its simplicity, blind dereverberation is a challenging problem due mainly to the characteristics of room impulse responses. For a general overview and discussion of relatively recent methods, we refer to the edited book \cite{naylor}. Here, we focus on a successful class of methods that can be collected under the heading `multichannel linear prediction' (MCLP), as they are more directly related to the content of our work. These methods usually fit in a schema as follows. Given multiple observations $y_i$, for $i=1,2,\ldots,M$, one of the observations is chosen as the reference. Then, this reference observation is linearly predicted by the delayed version of all of the observations. Once the linear prediction coefficients are determined, this information has been used differently in prior work. For instance, in \cite{kin09p534}, the authors use the coefficients to estimate the power of late reverberations, which are subsequently reduced by spectral subtraction. In more recent work \cite{nak10p717,juk16,juk15p509}, the residual of linear prediction is taken as the dereverbed audio signal. Prior work also differentiates based on the criteria used for determining the weights used in linear estimation. In \cite{kin09p534}, the authors propose to select the linear estimation weights so as to minimize the energy of the residual. In \cite{nak10p717}, the source signal is modeled as a time-varying Gaussian process and the weights are selected so as to maximize the likelihood of such a model. More recently, \cite{juk16, juk15p509} employ a sparsity- based model in the time-frequency domain and select the weights so that the residual minimizes a sparsity promoting cost function.

An interesting approach that employs MCLP  to achieve suppression of late reverberation and source separation is presented in \cite{tog13p369}. The authors use MCLP to suppress late reverberation, using a statistical framework similar to that in \cite{nak10p717}. In contrast to \cite{nak10p717}, a time-varying scenario with multiple sources is considered and the MCLP weights are determined along with the statistics of the late reverberant part of the sources in an iterative manner. 

\subsection*{Contribution}

Our proposal in this paper is of a complementary nature and is likely to improve the performance of existing separation algorithms. We present an inventive way to employ MCLP as a pre-processing step for blind source separation using a microphone array. We show theoretically that, under certain assumptions, such a pre-processing step converts the reverberant source separation problem into a non-reverberant one, which in turn can be solved by existing methods such as \cite{sou13p913,par02p352,valin04IROS} outlined above. Such an approach is inherently different than previous uses of MCLP which aim to dereverb a single source as in \cite{kin09p534,nak10p717,juk16,juk15p509}. Our proposal also differs from that of \cite{tog13p369} because \cite{tog13p369} requires the source statistics to cancel the late reverberation and thus resorts to a joint estimation of the linear prediction weights and source statistics. This approach in \cite{tog13p369} requires an iterative procedure and is handled with an EM algorithm. In contrast, we propose to employ the MCLP step once, in order to reduce the reverberant problem to an approximately non-reverberant one. 

In addition to  theoretical justification, we demonstrate the validity of the claims using real recordings. For the sake of completeness, we include a description of GSS along with a simple post-filter, to be applied following the MCLP step, in order to achieve separation. However, after MCLP, other methods that assume short RIRs, or operate under a free propagation model may also be used, as noted above.

\subsection*{Notation}
Sequences and continuous functions of time are denoted with a small letter as in $h(n)$ or $h(t)$, respectively. The DFT and STFT of sequences are denoted with a capital letter as $H(\omega)$ and $H(s,\omega)$ respectively -- which one is referred to will be clear from the context. The $z$-transform of sequences are denoted with calligraphic letters as in $\mathcal{H}(z)$. We refer to $\mathcal{H}(z)$ as `causal' if it is a polynomial of $z^{-1}$. For a causal $\mathcal{H}(z)$, we denote the constant $h(0)$ also as $\mathcal{H}(\infty)$.

Throughout the paper, $K$ denotes the number of sources and $N$  denotes the number of microphones in the array.
\section{Description of the Model}
Throughout the paper, we are interested in far-field processing. Thus, the distance of the source to the microphones is assumed to be significantly greater than the array size. Our proposed formulation/algorithm is in the STFT domain. However, for simplicity of exposition, we first introduce a model in the continuous-time domain and then present an approximation of this model in the STFT domain.

\subsection{A Time-Domain Observation Model}\label{sec:tdmodel}
We introduce the time-domain model in three stages, where complexity increases with each stage. We start with the case of a single source, single microphone. In the second stage, we consider the single source, multiple microphone case. Finally, we discuss the multiple source, multiple microphone case, which is the main scenario of interest in this paper.
\subsubsection{Single Source, Single Microphone}
For a given microphone-source pair, let $x(t)$ denote the signal produced by the source. The observation $y(t)$ can be approximated as,
\begin{equation}\label{eqn:RIR}
y(t) \approx x(t) \ast \tilde{h}(t),
\end{equation}
where $\tilde{h}(t)$ is the room impulse response (RIR) for the given source-microphone positions. Now suppose that the distance from the source to the microphone is $D$ meters and the speed of sound in the environment is $c$ m/sec. Then, the acoustic signal emitted by the source reaches the microphone after a delay of $D/c$ seconds. We declare the \emph{distance compensated room impulse response} (DCIR) to be 
\begin{equation}
h(t) = \tilde{h}\bigl(t + D/c \bigr).
\end{equation}
Notice that for practical purposes, we can work with the DCIR instead of the RIR for a single microphone-source setup. For multiple microphones or sources, we need to take into account relative distances, as discussed next. 

\subsubsection{Single Source, Multiple Microphones}
Consider   a setup consisting of a microphone array and a source as shown in Fig.~\ref{fig:setup}. Suppose we set the center of the array as the origin for a coordinate system. For this origin, let the vector $p_i$ denote the position of the the $i\thh$ microphone. Also, let $u_x$ denote the unit vector in the direction of the source. Finally, let $d_i = \langle p_i, u_x\rangle$. Note that $d_i$ denotes the coordinate of the projection of the $i\thh$ microphone's position to the subspace spanned by $u_x$.  The induced time-delay is $\tau_i = - d_i / c$. Notice that $\tau_i$ may be positive or negative, i.e., may represent an advance or delay. Finally, let the DCIR from the source to the $i\thh$ microphone be denoted as $h_i(t)$. Then, we model the observation at the $i\thh$ microphone, namely $y_i$, as,
\begin{equation}
y_i(t) \approx x(t) \ast h_i(t - \tau_i) = x(t - \tau_i) \ast h_i(t).
\end{equation}
The far field assumption amounts to the following : If the distance of the source to the center of the array is $D$, then the distance of the source to the $i\thh$ microphone is approximately $D - d_i$. Under the far-field assumption, the described model preserves the ability of the array to differentiate different directions. We next add another layer to this model by considering multiple sources.

\subsubsection{Multiple Sources, Multiple Microphones}

Suppose now that there are $K$ sources and $N$ microphones. Let the $k\thh$ source signal be denoted by $x_k(t)$. Also, let the unit vector in the direction of the $k\thh$ source be denoted by  $u_k$. For the $i\thh$ microphone, positioned at $p_i$, we set $\tau_{i,k} = \langle p_i, u_k \rangle / c$. Finally, let $h_{i,k}(t)$ denote the DCIR between the $k\thh$ source and the $i\thh$ microphone. Then, we model the observation at the $i\thh$ microphone as,
\begin{equation}\label{eqn:tdmodel}
y_i(t) \approx \sum_{k = 1}^K x_k(t) \ast h_{i,k}(t - \tau_{i,k}) = \sum_{k = 1}^K x_k(t - \tau_{i,k}) \ast h_{i,k}(t)
\end{equation}
Notice that, the absolute delay from the $k\thh$ source is not included in the model. However, this is not a careless omission, because this delay is in fact not known and is not estimated in practice. Therefore, the described model is suitable for practical far-field processing scenarios.

\begin{figure}
\centering
\includegraphics[scale=1]{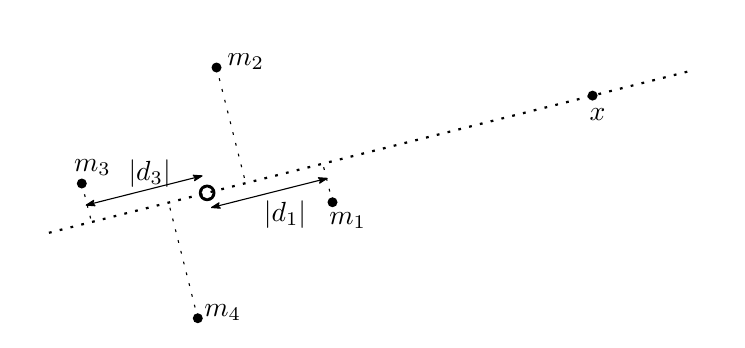}
\caption{Microphones and the source are denoted by $m_i$, and $x$ respectively. The center of the array, marked with a circle, is taken to be the origin. The position vector for the $i^{\text{th}}$ microphone is projected onto the space spanned by the unit vector pointing in the direction of the source to obtain the relative coordinate $d_i$. Note that $d_1$, $d_2$ (not shown) are positive whereas $d_3$, $d_4$ (not shown) are negative. \label{fig:setup}}
\end{figure}

\subsection{An Approximate Model in the STFT Domain}
We work in the STFT domain for reasons to be briefly clarified in the sequel. We are specifically interested in the representation of an LTI system in the STFT domain. Consider a convolution equation as,
\begin{equation}
y(n) = x(n) \ast h(n).
\end{equation}
Let $Y(s,\omega)$, $X(s,\omega)$ denote the STFTs of $y$ and $x$ respectively, obtained using a window $w(n)$. We differentiate between two cases, following the discussion in \cite{kow10p818} (see Lemma~1 in \cite{kow10p818}).
\begin{enumerate}[(I)]
\item\label{item:contained} If the support of $h$ is contained in the support of $w$, then
\begin{equation}
Y(s,\omega) \approx X(s,\omega)\,H(\omega),
\end{equation}
where $H(\omega)$ denotes the DFT of $h$.
\item If the support of $h$ is significantly larger than the support of $w$, then
\begin{equation}
Y(s,\omega) \approx X(s,\omega) \overset{s}{\ast} \tilde{H}(s,\omega),
\end{equation}
for some $\tilde{H}(s,\omega)$, where $\overset{s}{\ast}$ denotes convolution with respect to the time variable $s$, defined as
\begin{equation}
X(s,\omega) \overset{s}{\ast} \tilde{H}(s,\omega) = \sum_l X(s - l,\omega)\,\tilde{H}(l,\omega).
\end{equation}
We remark that $\tilde{H}(s,\omega)$ is not necessarily the STFT of $h$, but is chosen so as to make the approximation as close as possible -- see \cite{rei02p730} for a discussion.
\end{enumerate}

We remark that, as a special case of \eqref{item:contained}, if $\tau$ is small enough, then the STFT of $x(n - \tau)$ can be approximated as $X(s,\omega)\,\exp(-j\omega\,\tau)$. Therefore, the model introduced in \eqref{eqn:tdmodel} can be approximated in the STFT domain as
\begin{equation}\label{eqn:STFTdomain}
Y_i(s,\omega) \approx \sum_{k=1}^K \,  X_k(s,\omega)\,\overset{s}{\ast}\, \Bigl( \tilde{H}_{i,k}(s,\omega)\,\exp(-j\omega\,\tau_{i,k}) \Bigr).
\end{equation}

Such a model is preferred over the time domain model because prior or penalty functions that enforce (structured) sparsity of the audio signals are typically defined in the STFT (or some other time-frequency) domain. If we use the time domain convolutional model  \eqref{eqn:tdmodel} in an iterative algorithm, we have to work in the time domain to undo the effects of reverberation, and the STFT domain to take advantage of the signal prior. Switching from the time domain representation to the STFT representation at each iteration increases the computational cost significantly. Working with a convolutional model directly in the STFT domain, such as \eqref{eqn:STFTdomain}, allows to carry out the iterations in the STFT domain, eliminating the need to apply the forward and inverse STFT operations at each iteration.

\section{Multi-Channel Linear Prediction for Dereverberation}
In this section, we discuss the implications of using multi-channel linear prediction (MCLP) as a preliminary step for the source separation problem. We derive conditions under which MCLP transforms the original blind separation/dereverberation problem to a non-blind one. 

We start the discussion with a single source and then consider extensions to multiple sources.

\subsection{MCLP for a Single Source}\label{sec:MCLPss}
Consider a signal $x$ that is observed via different channels as
\begin{equation}
y_i = x \ast h_i, \text{ for } i=1,2,\ldots, N,
\end{equation}
where $h_i$ denotes the impulse response of the $i\thh$ channel. 
In the $z$-domain, this can be written as,
\begin{equation}
\begin{bmatrix}
\z{Y}_1(z) \\ \vdots \\ \z{Y}_N(z)  
\end{bmatrix}
=
\begin{bmatrix}
\z{H}_1(z) \\ \vdots \\ \z{H}_N(z)  
\end{bmatrix}
\z{X}(z).
\end{equation}
We would like to find causal filters $\z{G}_1(z)$, \ldots, $\z{G}_N(z)$ such that 
\begin{equation}\label{eqn:MCLPdesired}
\begin{bmatrix}
\z{G}_1(z) & \ldots & \z{G}_N(z)  
\end{bmatrix}\,
\begin{bmatrix}
\z{H}_1(z) \\ \vdots \\ \z{H}_N(z)  
\end{bmatrix}
 = 1.
\end{equation}
MINT \cite{miy88p145} ensures that such filters exist provided that $\mathcal{H}_i$'s are co-prime, i.e., they do not have a common non-constant factor.
Even though it is relieving to know that such filters exist, finding the filters is still a challenging problem when $h_i$'s are not known. MCLP addresses this challenge.

In MCLP, the idea is to estimate $x$ by subtracting from $y_{r}$ (for a chosen reference index $r$) delayed weighted combinations of $y_i$'s.
This is equivalent to using in \eqref{eqn:MCLPdesired} a set of filters of the form
\begin{align}\label{eqn:MCLPfilters}
\begin{split}
\z{G}_r(z) &= 1 - z^{-d}\,\z{U}_r(z) \text{ for a reference index }r, \\
\z{G}_i(z) &= -z^{-d}\,\z{U}_i(z) \text{ for }i=1,\ldots,N,\,i\neq r, 
\end{split}
\end{align}
for a positive integer $d$ (determining the amount of delay) and causal $\z{U}_k$'s. 
In \cite{juk16}, the selection of the filters is guided by the assumption that the original $x$ is sparse. More specifically, suppose $y_i(s)$ denotes the 1D time-series derived from a frequency band of the observations, that is $y_i(s) = Y_i(s,\omega)$ for some $\omega$. The filters $u_i$ are found by solving the following minimization problem. 
\begin{equation}
\min_{u_1,\ldots,u_N} P\left(y_r(s) - \sum_{i=1}^N \sum_{l = 0}^L y_i(s-d - l)\, u_i(l) \right),
\end{equation}
where $P(\cdot)$ is a sparsity promoting function. We will have more to say about $P$, but first we would like to stress a point about the choice \eqref{eqn:MCLPfilters}.
MINT ensures the existence of dereverberating filters but does not guarantee that they can be of the form \eqref{eqn:MCLPfilters}. Below, we propose modified conditions which ensure the existence of such $G_i$'s. Another issue is the role of the reference index $r$. It turns out that repeating the linear prediction for different $r$ values leads to an interesting outcome, summarized in Prop.~\ref{prop:Observations} below. In order to derive Prop.~\ref{prop:Observations}, we need an auxiliary result (Prop.~\ref{prop:modMINT}) and an assumption.
\begin{prop}\label{prop:modMINT}
Suppose that the $N-1$ polynomials  $\z{H}_i$, $i \in \{ 1,\ldots, N \} \setminus \{r\}$, are co-prime. Then, for $d = 1$, there exist a set of filters of the form \eqref{eqn:MCLPfilters} such that 
\begin{equation}\label{eqn:propMCLP}
\begin{bmatrix}
\z{G}_1(z) & \ldots & \z{G}_N(z)  
\end{bmatrix}\,
\begin{bmatrix}
\z{H}_1(z) \\ \vdots \\ \z{H}_N(z)  
\end{bmatrix}
 = c,
\end{equation}
for a constant $c$.
Further, for a set of causal filters of the form \eqref{eqn:MCLPfilters}, if \eqref{eqn:propMCLP} holds, then $c  = \z{H}_r(\infty)$.
\begin{proof}
See Appendix~\ref{app:modMINT}.
\end{proof}
\end{prop}

Consider now the application of this result to the model \eqref{eqn:STFTdomain}. Assume that there is a single source, i.e., $K=1$. For a fixed frequency  $\omega$, let 
\begin{equation}
h_i(s) = \tilde{H}_{i,1}(s,\omega)\,\exp(-j\omega\,\tau_{i,1}). 
\end{equation}
Under the hypotheses of Prop.~\ref{prop:modMINT}, if, for $d=1$, MCLP can recover a set of filters satisfying \eqref{eqn:propMCLP}, then Prop.~\ref{prop:modMINT} implies that,
\begin{equation}\label{eqn:c}
\hat{y}_r(s) = h_r(0)\,\exp(-j\omega\,\tau_{r,1})\,x(s).
\end{equation}
Now, if further,
\begin{equation}\label{eqn:assumption}
h_i(0) = h_{i'}(0) \text{ for any distinct pair }(i,i'),
\end{equation}
%by employing MCLP with different reference indices, we obtain
then,
\begin{equation}\label{eqn:cprime}
\hat{y}_r(s) = c' \,\exp(-j\omega\,\tau_{r,1})\, x(s),
\end{equation}
for some constant $c'$, independent of $r$ (actually, $c' = h_i(0)$, for any $i$). We remark that the difference of \eqref{eqn:c} and \eqref{eqn:cprime} is that the rhs of \eqref{eqn:cprime} depends only on the direction of arrival, whereas this is not the case in \eqref{eqn:c}.
In other words, MCLP eliminates reflections from different directions. 

The assumption in \eqref{eqn:assumption} states that, around the origin, distance compensated impulse responses are similar. This assumption is key to the development in this paper and we state it explicitly for later reference. We will  verify this assumption experimentally in Section~\ref{sec:exp}.
\begin{assm}\label{assm:main}
For a given source - microphone array pair, let $\tilde{H}_i(s,\omega)$ denote the STFT representation of the distance compensated impulse response for the $i\thh$ microphone. Then,
\begin{equation}
\tilde{H}_i(0,\omega) = \tilde{H}_{i'}(0,\omega),
\end{equation}
for $i \neq i'$ and all $\omega$.
\end{assm}

We state the foregoing observations as a proposition.
\begin{prop}\label{prop:Observations}
For a fixed frequency  $\omega$, let $h_i(s) = \tilde{H}_{i,1}(s,\omega)\,\exp(-j\omega\,\tau_{i,1})$. Assume also that the collection of polynomials $\z{H}_i$, $i \in \{ 1,\ldots, N \} \setminus \{r\}$ are co-prime. Then, for $d = 1$, there exist a set of filters of the form \eqref{eqn:MCLPfilters} such that 
\begin{equation}\label{eqn:propMCLP1}
\sum_i g_i(n) \ast h_i(n) = c_r\, \delta(n)
\end{equation}
for a constant $c_r$.

Further, for a set of filters of the form \eqref{eqn:MCLPfilters}, if \eqref{eqn:propMCLP1} holds, and Assumption~\ref{assm:main} is in effect, then $c_r  = c'\,\exp(-j\omega\tau_{r,1})$ for a constant $c'$, independent of $r$. \qed
\end{prop}

Next, we  extend the discussion to the case of multiple sources.

\subsection{MCLP for Multiple Sources}\label{sec:MultiSourceMCLP}

Consider an observation model with $K$ sources $x_1$,\ldots $x_K$ and $N>K$ observations defined as,
\begin{equation}\label{eqn:observation}
y_i = \sum_{k=1}^K x_k \ast h_{i,k}, \text{ for } i=1,2,\ldots, N.
\end{equation}
In the $z$-domain, this can be written as,
\begin{equation}\label{eqn:modelMC}
\begin{bmatrix}
\z{Y}_1(z) \\ \vdots \\ \z{Y}_N(z)  
\end{bmatrix}
=
\underbrace{ \begin{bmatrix}
\z{H}_{1,1}(z) & \ldots &\z{H}_{1,K}(z) \\ & \ddots&  \\ \z{H}_{N,1}(z)  & \ldots & \z{H}_{N,K}(z)
\end{bmatrix}}_{\bo{H}}
\begin{bmatrix}
\z{X}_1(z) \\
\ldots \\
\z{X}_K(z) 
\end{bmatrix}
\end{equation}
For this model, MINT also ensures that we can find filters $\z{G}_{k,i}$, for $k=1,\ldots,K$, $i=1,\ldots, N$ such that
\begin{equation}\label{eqn:multisource}
\begin{bmatrix}
\z{G}_{1,1}(z) & \ldots & \z{G}_{1,N}(z) \\ & \ddots&  \\ \z{G}_{K,1}(z)  & \ldots & \z{G}_{K,N}(z)
\end{bmatrix}
\,\begin{bmatrix}
\z{H}_{1,1}(z) & \ldots &\z{H}_{1,K}(z) \\ & \ddots&  \\ \z{H}_{N,1}(z)  & \ldots & \z{H}_{N,K}(z)
\end{bmatrix}
= I,
\end{equation}
provided that the (Smith) canonical form of $\bo{H}$ (see Appendix~\ref{app:Smith} or \cite{Gantmacher1}, Sec.VI.2, Defn.3) has a non-zero constant (i.e., no polynomial higher than zero order) diagonal. However, unlike the single source case, we cannot use filters of the form \eqref{eqn:MCLPfilters} to achieve \eqref{eqn:multisource}. We have the following extension in this case.
\begin{prop}\label{prop:MCLP2}
For a fixed reference microphone index $r \in \{1,\ldots,N\}$, let $\bar{\bo{H}}_r$ be the matrix obtained by removing the $r\thh$ row of $\bo{H}$. If the greatest common divisor of all the $K$-minors of $\bar{\bo{H}}_r$ is a non-zero constant, then, for $d = 1$, there exist a set of filters of the form \eqref{eqn:MCLPfilters} such that 
\begin{equation}\label{eqn:propMCLP2}
\begin{bmatrix}
\z{G}_1(z) & \ldots & \z{G}_N(z)  
\end{bmatrix}\,
\begin{bmatrix}
\z{H}_{1,1}(z) & \ldots & \z{H}_{1,K}(z) \\ & \ddots&  \\ \z{H}_{N,1}(z)  & \ldots & \z{H}_{N,K}(z)
\end{bmatrix}
 = c^{(r)},
\end{equation}
for a constant vector $c^{(r)}$ (i.e., containing no polynomial of order greater than zero).

Further, for a set of filters of the form \eqref{eqn:MCLPfilters}, if \eqref{eqn:propMCLP2} holds, then the components of $c$ satisfy,
\begin{equation}\label{eqn:propcr}
c^{(r)}_k = \mathcal{H}_{r,k}(\infty), \text{ for }k=1,2,\ldots, K.
\end{equation}
\begin{proof}
See Appendix~\ref{app:MCLP2}.
\end{proof}
\end{prop}

\begin{remark}\label{rem:equivalence}
For an $N\times K$ matrix $\bo{H}$ with $K\leq N$, it can be shown that, the greatest common divisor of all the $K$-minors of $\bo{H}$ is a non-zero constant if and only if the Smith canonical form of $\bo{H}$ has a  non-zero constant diagonal. This equivalence is shown in Prop.~\ref{prop:equivalence} in Appendix~\ref{app:Smith}. Therefore, we can replace the hypothesis in  Prop.~\ref{prop:MCLP2} with an alternative, stated in terms of the Smith canonical form of $\bar{\bo{H}}_r$.
\end{remark}

When we apply this result to our observation model with $K$ sources, we obtain the following proposition.
\begin{prop}\label{prop:main}
For a fixed frequency  $\omega$, let 
\begin{equation}\label{eqn:prophi}
h_{i,k}(s) = \tilde{H}_{i,k}(s,\omega)\,\exp(-j\omega\,\tau_{i,k}), 
\end{equation}
for $i\in\{1,\ldots,N\}$, $k\in \{1,\ldots, K\}$. Also, let $\bo{H}$ be defined as in \eqref{eqn:modelMC}, and $\bar{\bo{H}}_r$ be defined as in Prop.~\ref{prop:MCLP2}. Finally, let $y_i(s) = Y_i(s,\omega)$.

If the greatest common divisor of all the $K$-minors of $\bar{\bo{H}}_r$ is a non-zero constant, then, for $d = 1$, there exist a set of filters of the form \eqref{eqn:MCLPfilters} such that 
\begin{equation}\label{eqn:propMCLP3}
\sum_{i=1}^N g_i(s) \ast h_{i,k}(s) = c^{(r)}_k\, \delta(s)
\end{equation}
for some constants $c^{(r)}_k$.

Further, for a set of filters of the form \eqref{eqn:MCLPfilters}, if \eqref{eqn:propMCLP3} holds, and Assumption~\ref{assm:main} is in effect, 
then we have
\begin{equation}\label{eqn:propMCLPDOA}
\sum_{i=1}^N\,g_i(s) \ast y_i(s) = c'\,\sum_{k=1}^K\,x_k(s)\,\exp(-j\,\omega\,\tau_{r,k}),
\end{equation}
where $c'$ is a constant independent of $r$.
\begin{proof}
See Appendix~\ref{app:main}.
\end{proof}
\end{prop}

This proposition suggests that, in a scenario with multiple sources, MCLP, if successful, suppresses reflections and preserves the direct signal even when there exist multiple sources from different directions. Therefore, after the MCLP step, the recordings appear as if they are recorded with the same microphone array, but in a non-reverberant environment.
This in turn allows to read off the directions of the sources with a simple procedure like MUSIC \cite{sch86p276}. Once the source directions are determined, we can then achieve separation by using \emph{known} manifold vectors of the microphone array. 

\section{Implementation Details}
The discussion after Prop.~\ref{prop:main} sets the stage for the proposed framework. The main building blocks of the framework are shown in Fig.~\ref{fig:framework}. We first input the $N$ observations $y_1 , \ldots,  y_N$, recorded in a reverberant environment with unknown reverberation, to an MCLP block. This block applies MCLP for every possible choice of the reference microphone. The output of this block is $N$ dereverberated signals $\hat{y}_1,\ldots , \hat{y}_N$. These signals are input to a DOA estimation algorithm (such as MUSIC \cite{sch86p276}) and the directions of the $K$ sources are determined. Then these estimated directions are used to separate the sources in $\hat{y}_i$ using a non-blind source separation method. The estimated sources are denoted as $\hat{x}_i$ in Fig.~\ref{fig:framework}. In the following, we briefly present the algorithm we used for a practical implementation of the components. However, we remark that following MCLP, the problem turns into a non-reverberant source separation problem (with known array geometry) and one can use any method of choice for this task.

\begin{figure}
\centering
\includegraphics[scale=1]{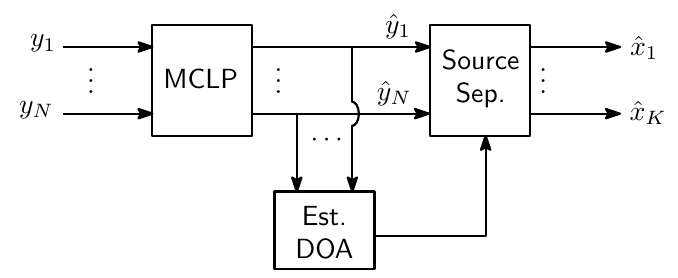}
\caption{The building blocks of the proposed framework. The input to the system is $y_1$, \ldots $y_N$, which are $N$ reverberant observations of $K$ sources. 
Source directions are estimated under the assumption of known array geometry. The estimated source directions are used in a non-blind source-separation algorithm. \label{fig:framework}}
\end{figure}

\subsection{Multi-Channel Linear Prediction}
We use the formulation given in Sec.~\ref{sec:MultiSourceMCLP}. We work in the STFT domain and process the frequency channels separately. For a fixed value of $\omega$, we set $y_i(s) = Y_i(s,\omega)$. Then for each value of $r \in \{1,\ldots,N\}$, the dereverbed estimate  $\hat{y}_r$ is defined as,
\begin{equation}
\hat{y}_r(s) = y_r(s) - \sum_{i=1}^N \sum_{l = 0}^{L-1} y_i(s-d - l)\, \hat{u}^{(r)}_i(l).
\end{equation}
where $\hat{u}^{(r)}_i$ denote filters determined by solving
\begin{multline}\label{eqn:MCLPform}
\left(\hat{u}^{(r)}_1, \ldots, \hat{u}^{(r)}_N \right) \\= \arg \min_{u_1,\ldots,u_N}  
P\left(y_r(s) - \sum_{i=1}^N \sum_{l = 0}^L y_i(s-d - l)\, u_i(l) \right).
\end{multline}
We use a weighted $\ell_1$ norm as the sparsity promoting function $P(\cdot)$. More precisely, for $u \in \mathbb{C}^n$, we set 
\begin{equation}\label{eqn:wl1}
P(u) = \|u\|_w = \sum_{i}^{n} w_i \, |u_i|,
\end{equation}
where $w \in \mathbb{R}_+^n$ is a given weight vector. We update the weight vector $w$ as the iterations progress, to improve the sparsity promoting property of $P$. 

We remark that for MCLP, a similar formulation  has been discussed in \cite{juk16}, along with an algorithm based on the alternating direction method of multipliers (ADMM) algorithm \cite{boy11p1}. We employ the Douglas-Rachford (DR) algorithm \cite{lio79p964, combettes_chp} for solving this problem. To that end, we first define the matrices $X^{(i)}$ for $i\in \{1,\ldots,N\}$ such that
\begin{equation}
X^{(i)}_{s,l} = y_i(s - d - l), \text{ for }l=0,1,\ldots,L-1,
\end{equation}
and as $s$  ranges over the domain of the observation $y_r$.
Notice that $X^{(i)}$ is Toeplitz and has $L$ columns and as many rows as the number of samples in $y_r$. By concatenating $X^{(i)}$'s, we construct the data matrix $X$ as,
\begin{equation}
X = \begin{bmatrix}
X^{(1)} & \cdots & X^{(n)}
\end{bmatrix}.
\end{equation}
Similarly, we construct the variable vector $u$ as 
\begin{equation}
u = \begin{bmatrix} u^{(1)} \\ \vdots \\ u^{(N)} \end{bmatrix},
\end{equation}
where $u^{(i)}$ denotes the column vector with $u^{(i)}_k = u_i(k)$. Also, let $y^{(r)}$ denote the column vector with $y^{(r)}_k = y_r(k)$. 
Using these, we can rewrite the problem \eqref{eqn:MCLPform} as,
\begin{equation}
\min_{z,u} \| z \|_w, \text{ subject to } z = y^{r} - X\,u
\end{equation}
We split variables and introduce $z$ in order to make the formulation feasible for employing DR.

For this formulation, the Douglas-Rachford algorithm is of the form outlined in Algorithm~\ref{algo:DR}. We remark that in this algorithm, at each $M$ iterations, the weights $w$ used to compute $\|\cdot\|_w$ are updated, so as to improve the sparsity inducing property of the penalty function, following \cite{juk16}. We refer to \cite{juk16} for more intricate update strategies.

\begin{algorithm}\caption{Multi-Channel Linear Prediction}\label{algo:DR}
\begin{algorithmic}[1]
\Require Input Data $y^{(r)}$, $X$, a constant for weight updates $\epsilon>0$
\State Initialize $w_i \gets  1$ for all $i$
\Repeat
\For{m = 1 : M} \Comment Update $w$ every $M$ iterations
\State $\tilde{z}_i \gets 2\,\soft(z,\,\alpha\, w) - z$
\State $\tilde{u} \gets (I + X^T\,X)^{-1}\,\bigl( u + X^T\,(y^{r} - \tilde{z}) \bigr)$
\State $\tilde{z} \gets (1-\lambda)\,z + \lambda\,\bigl(2\,(y - X\,\tilde{u}) - \tilde{z} \bigr)$
\State $u \gets (1-2\lambda)\,u + 2\lambda\,\tilde{u}$
\EndFor
\State $\tilde{z} \gets \soft(z,\,\alpha\, w)$
\State $w_i \gets \epsilon / (\epsilon + |z_i|)$ for all $i$ \Comment{update  $w$}
\Until{convergence}
\State \Return $\hat{y}_r = \soft(z,\,\alpha\, w)$
\end{algorithmic}
\end{algorithm}

\subsection{DOA Estimation}\label{sec:DOA}
Because audio is a wideband signal, narrowband DOA estimation techniques like MUSIC \cite{sch86p276}, or ESPRIT \cite{roy89p984} are not directly applicable. However, simple extensions of MUSIC to wideband scenarios exist and we employ one of them. In a nutshell, the idea is to estimate the DOA for each frequency separately and average the estimates. This approach is  referred to as the `incoherent signal subspace method' (ISSM) in the recent literature \cite{yoo06p977} -- see \cite{wax84p817,su83p502} for early work. Despite its simplicity, ISSM is noted to perform well in high SNR \cite{yoo06p977}.

Subspace methods for DOA estimation make use of two essential components, namely the covariance matrix of the data and the manifold vector, determined by the array geometry. For a fixed frequency $\omega$, and time instance $s$ (in the STFT domain), we estimate the correlation matrix $R_Y(s,\omega)$ as
\begin{equation}
R_Y(s,\omega) = (1-\eta)\,M(s,\omega) + \eta\,R_Y(s-1,\omega),
\end{equation}
where 
\begin{equation}
M(s,\omega) = 
\begin{bmatrix}
\hat{Y}_1(s,\omega) \\
\vdots \\
\hat{Y}_N(s,\omega)
\end{bmatrix}
\,
\begin{bmatrix}
\hat{Y}_1^*(s,\omega) & \cdots & \hat{Y}_N^*(s,\omega)
\end{bmatrix}.
\end{equation}
In this definition, $\eta \in [0,1]$ is a smoothing (or forgetting) factor justified by the assumption that the sources move slowly.

Let us now define the manifold vector. Recall from Sec.~\ref{sec:tdmodel} that $p_i$ denotes the position of the $i\thh$ microphone with respect to the center of the array. Also, let $u_{\theta}$ denote the unit vector in the direction $\theta$. Finally, let $c$ denote the speed of sound, $f_s$ denote the sampling frequency and 
\begin{equation}
\tau_{i,\theta} = - \frac{\langle p_i, u_{\theta}\rangle}{c} \,f_s
\end{equation}
denote the relative delay (in samples) of the $i\thh$ microphone  for a source in the direction $\theta$. The manifold vector is defined as 
\begin{equation}
a_{\theta}(\omega) = \begin{bmatrix}
\exp\bigl(-j \tau_{1,\theta}\,\omega \bigr) \\
\vdots\\
\exp\bigl(-j \tau_{N,\theta}\,\omega \bigr)
\end{bmatrix}.
\end{equation}
Assuming there are $K$ sources, let $C_K(s,\omega)$ denote the $(N-K) \times N$ unitary matrix whose columns are the eigenvectors of $R_Y(s,\omega)$ corresponding to the  smallest $N-K$ eigenvalues. The DOA function is defined as 
\begin{equation}\label{eqn:DOAfunct}
D(s,\theta) = \left( \sum_{\omega} a_{\theta}^*(\omega)\, C_K(s,\omega) \, a_{\theta}(\omega)\right)^{-1}.
\end{equation}
For each time instance $s$, the $K$ dominant peaks of $D(s,\theta)$ with respect to $\theta$ are taken to be the directions of the sources.
\subsection{Source Separation}\label{sec:SS}
For source separation, we use the geometric source separation (GSS) method \cite{par02p352,valin04IROS} complemented with a simple post-filter. We assume that the sources are stationary, so that the estimated directions are constant with respect to time. Extension to moving sources is straightforward, at least in principle, and will not be further discussed.

\subsubsection{Geometric Source Separation}\label{sec:GSS}
Let $\theta_i$, for $i=1,\ldots, K$ denote the directions of the sources determined by the method described in Sec.~\ref{sec:DOA}. For each $\omega$, we define a matrix $A(\omega)$ as,
\begin{equation}
A(\omega) = \begin{bmatrix}
a_{\theta_1}(\omega) & \cdots & a_{\theta_K}(\omega)
\end{bmatrix}
\end{equation}
We model the MCLP outputs as
\begin{equation}
\underbrace{\begin{bmatrix}
\hat{Y}_1(s,\omega) \\
\vdots \\
\hat{Y}_N(s,\omega)
\end{bmatrix}}_{\hat{Y}(s,\omega)} =
A(\omega)\,
\underbrace{\begin{bmatrix}
X_1(s,\omega) \\
\vdots \\
X_K(s,\omega)
\end{bmatrix}}_{X(s,\omega)}.
\end{equation}
This model is not exact because despite the significant dereverberation by  MCLP, $\hat{Y}$ still contains reflections other than the direct signal.

For this model, GSS proposes to estimate $X(s,\omega)$ as $\hat{W}(\omega)\,\hat{Y}(s,\omega)$, where the $K\times N$ weight matrix $\hat{W}(\omega)$ is determined by solving a minimization problem. To describe the problem, let us denote 
\begin{equation}
S(s,\omega) = W(\omega)\,Y(s,\omega),
\end{equation}
for a given weight matrix $W$. Also, let $R_S(\omega)$ denote the empirical autocorrelation matrix of $S(s,\omega)$, obtained by averaging $S(s,\omega)\,S^*(s,\omega)$ over the time variable $s$.
Given this notation, $\hat{W}(\omega)$ is defined as the solution of 
\begin{multline}
\min_{W(\omega)} \left\{ C\bigl(W(\omega)\bigr) = \frac{\gamma}{2}\,\bigl\| R_S(\omega) - \diag\bigl(R_S(\omega) \bigr) \bigr\|_F^2 \right. \\
\left. + \frac{1}{2}\,\bigl\| W(\omega)\,A(\omega) - I \bigr\|_F^2 \right\},
\end{multline}
where $\|\cdot\|_F$ denotes the Frobenius norm.

Thanks to the differentiability of both terms in the cost function, gradient descent can be applied for obtaining a stationary point. For a fixed $\omega$, we remark that the gradient \cite{bra83p11} of the complex valued $C(W)$ is given as \cite{par02p352},
\begin{equation}
\frac{\partial C(W)}{\partial W^*} = \gamma\,\bigl(R_S - \diag(R_S)\bigr)\,W\,R_{\hat{Y}} + (W\,A - I)\,A^*,
\end{equation}
where  $R_{\hat{Y}}$ denotes the empirical autocorrelation matrix of $\hat{Y}$, obtained by averaging $\hat{Y}^*(s,\omega)\,\hat{Y}(s,\omega)$ over the time variable $s$.

\subsubsection{Post-Filter}\label{sec:PF}

The GSS algorithm aims to achieve source separation by considering the best linear combination of the inputs. In practice, separation performance can be improved by post-filtering. Specifically, we view the separated signals as contaminated with noise and eliminate this noise using a soft-threshold, where the time-frequency varying threshold value depends on an estimate of the remaining noise variance for each time-frequency sample. 

More specifically, if $\hat{S}_1$, \ldots, $\hat{S}_K$ denotes the $K$ source estimates obtained by GSS, for each value of the triplet $(i,s,\omega)$, we regard $\hat{S}_i(s,\omega)$ as an unbiased estimate of $X_i(s,\omega)$ with variance $\sigma_i^2(s,\omega)$. The post-filtered estimate is formed as,
\begin{equation}
\hat{X}_i(s,\omega) = \soft\Bigl(\hat{S}_i(s,\omega), \alpha\,\sigma_i(s,\omega) \Bigr),
\end{equation}
where $\alpha$ is a parameter.

In order to realize such a post-filter, we need to estimate $\sigma_i^2(s,\omega)$. We  estimate $\sigma_i^2(s,\omega)$ as, 
\begin{multline}\label{eqn:sigma}
\frac{1}{N-1}\,\sum_{n=1}^N\, |W_{i,n}(s,\omega)| \\ \times  \left|\Bigl(\hat{Y}_n(s,\omega) - a_{\theta_i,n}(\omega) \,\hat{S}_i(s,\omega) \Bigr) \right|^2 ,
\end{multline}
where $a_{\theta_i,n}(\omega)$ denotes the $n\thh$ entry of $ a_{\theta_i}(\omega)$. While this expression may make intuitive sense, further justification is provided in Appendix~\ref{app:PF}. We also remark that this is a generalized version of the post-filter employed in \cite{bay15p272}.

\section{Experiments}\label{sec:exp}
Our primary claim in this paper is that an MCLP step suppresses reflections other than the direct source, under a certain assumption (Assumption~\ref{assm:main}). The experiments in this section are geared towards demonstrating the claims and verifying Assumption~\ref{assm:main}. There are three experiments. In the first, we verify Assumption~\ref{assm:main} using real impulse responses obtained in a classroom. In the second experiment, we demonstrate that an MCLP preprocessing step eliminates reflections other than the direct source, by showing how the DOA function changes after the MCLP step. Finally, the third experiment validates our claim about the effect of employing an MCLP preprocessing step in a source separation problem. 

All of the experiments are performed on real data. In order to obtain the data, we used a circular microphone array with eight omni-directional microphones. The configuration and the definition of DOA, namely $\theta$, are as shown in Fig.~\ref{fig:MicSetup}. The microphones are distributed uniformly on a circle of radius 5~cm.

In the experiments, we use an STFT with a smooth window of length 64~ms, and hop-size is 16~ms. The window is selected such that the STFT is self-inverting (i.e., a Parseval frame). The sampling frequency for the audio signals is 16 KHz. For the MCLP block, the parameters (see Algorithm~\ref{algo:DR}) are chosen as $d = 2$, $L=30$, $M = 50$, $\alpha = 0.05$, $\lambda = 0.5$. We run the algorithm for 2000 iterations.

Some of the audio signals used in the experiments, as well as those from experiments not contained in the paper can be listened to at ``http://web.itu.edu.tr/ibayram/MCLP/''.

\begin{figure}
\centering
\includegraphics[scale=1]{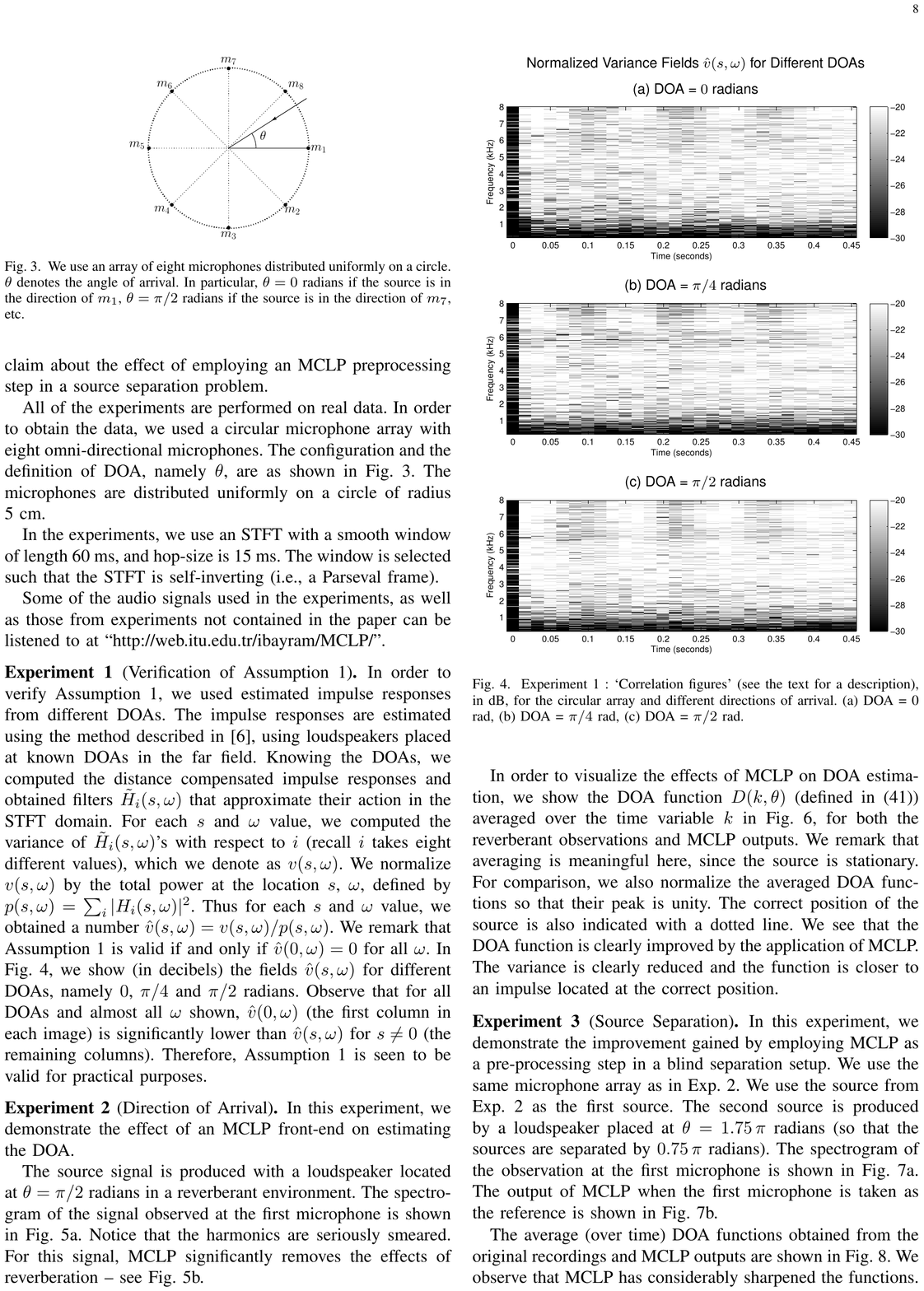}
\caption{We use an array of eight microphones distributed uniformly on a circle. $\theta$ denotes the angle of arrival. In particular, $\theta = 0$ radians if the source is in the direction of $m_1$, $\theta = \pi/2$ radians if the source is in the direction of $m_7$, etc. \label{fig:MicSetup}}
\end{figure}

\begin{experiment}[Verification of Assumption~\ref{assm:main}]\label{exp:assm}

In order to verify Assumption~\ref{assm:main}, we used estimated impulse responses from different DOAs. The impulse responses are estimated using the method described in \cite{far00aes}, using loudspeakers placed at known DOAs in the far field. Knowing the  DOAs, we computed the distance compensated impulse responses and obtained filters $\tilde{H}_i(s,\omega)$ that approximate their action in the STFT domain. 
For each $s$ and $\omega$ value, we computed the variance of $\tilde{H}_i(s,\omega)$'s with respect to $i$ (recall $i$ takes eight different values), which we denote as $v(s,\omega)$. We normalize $v(s,\omega)$ by  the total power at the location $s$, $\omega$, defined by $p(s,\omega) = \sum_{i} |H_i(s,\omega)|^2$. Thus for each $s$ and $\omega$ value, we obtained a number $\hat{v}(s,\omega) = v(s,\omega) / p(s,\omega)$. We remark that Assumption~\ref{assm:main} is valid if and only if $\hat{v}(0,\omega) = 0$ for all $\omega$. In Fig.~\ref{fig:Corr}, we show (in decibels) the fields $\hat{v}(s,\omega)$ for different DOAs, namely $0$, $\pi/4$ and $\pi/2$ radians. Observe that for all DOAs and almost all $\omega$ shown, $\hat{v}(0,\omega)$ (the first column in each image) is significantly lower than $\hat{v}(s,\omega)$  for $s\neq 0$ (the remaining columns). Therefore, Assumption~\ref{assm:main} is seen to be valid for practical purposes.

\begin{figure}
\renewcommand{\sc}{0.5}
\centering 
\includegraphics[scale=1]{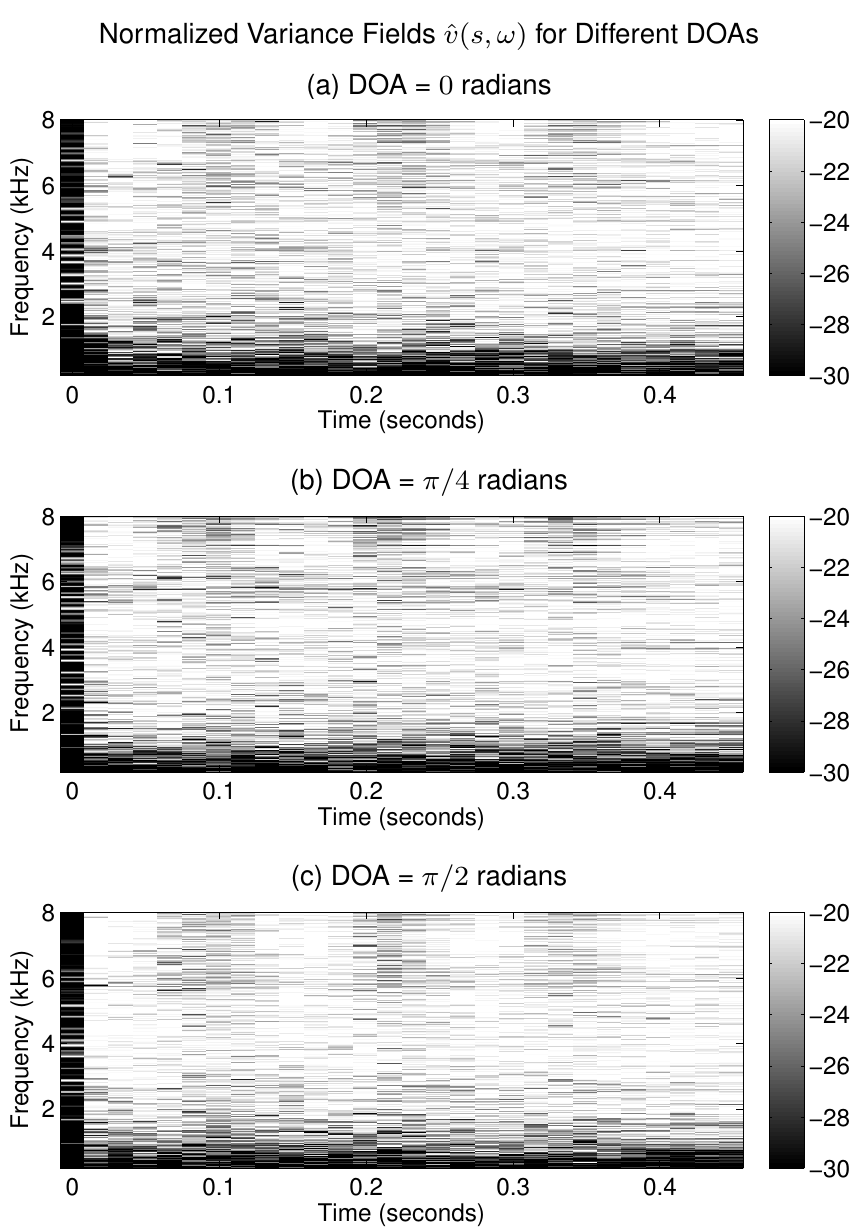}
\caption{Experiment~\ref{exp:assm} : `Correlation figures' (see the text for a description), in dB, for the circular array and different directions of arrival. (a) DOA = 0 rad, (b) DOA = $\pi/4$ rad, (c) DOA = $\pi/2$ rad. \label{fig:Corr}}
\end{figure}

\end{experiment}

\begin{experiment}[Direction of Arrival]\label{exp:DOA}

In this experiment, we demonstrate the effect of an MCLP front-end on estimating the DOA.

The source signal is produced with a loudspeaker located at $\theta = \pi/2$ radians in a reverberant environment. The spectrogram of the signal observed at the first microphone is shown in Fig.~\ref{fig:DOAspect}a. Notice that the harmonics are seriously smeared. For this signal, MCLP significantly removes the effects of reverberation -- see Fig.~\ref{fig:DOAspect}b.

\begin{figure}
\centering
\includegraphics[scale=1]{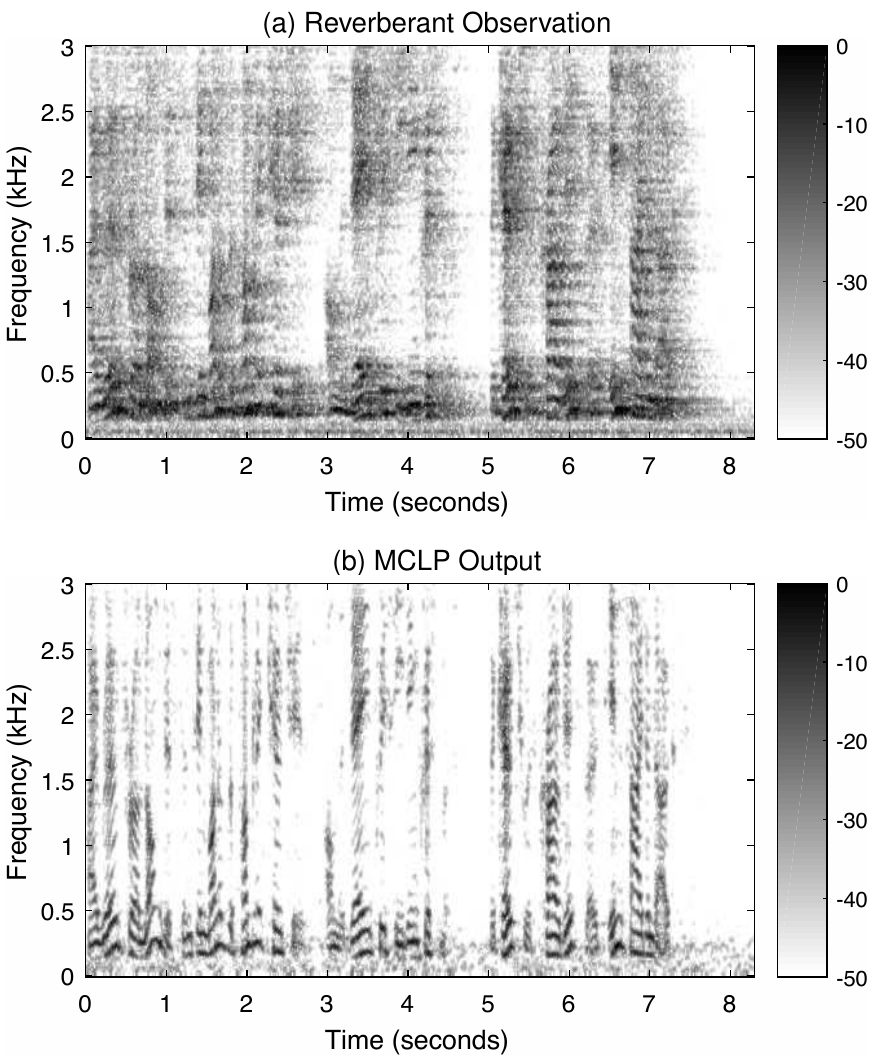}
\caption{Experiment~\ref{exp:DOA} : Spectrograms of the signals used in the experiment. (a) The observation at the first microphone. (b) MCLP output when the first microphone is used as the reference. \label{fig:DOAspect}}
\end{figure}

In order to visualize the effects of MCLP on DOA estimation, we show the DOA function $D(k,\theta)$  (defined in \eqref{eqn:DOAfunct}) averaged over the time variable $k$ in Fig.~\ref{fig:DOA}, for both the reverberant observations and MCLP outputs. We remark that averaging is meaningful here, since the source is stationary. For comparison, we also normalize the averaged DOA functions so that their peak is unity. The correct position of the source is also indicated with a dotted line. We see that the DOA function is clearly improved by the application of MCLP. The variance is clearly reduced and the function is closer to an impulse located at the correct position.

\begin{figure}
\centering
\includegraphics[scale=1]{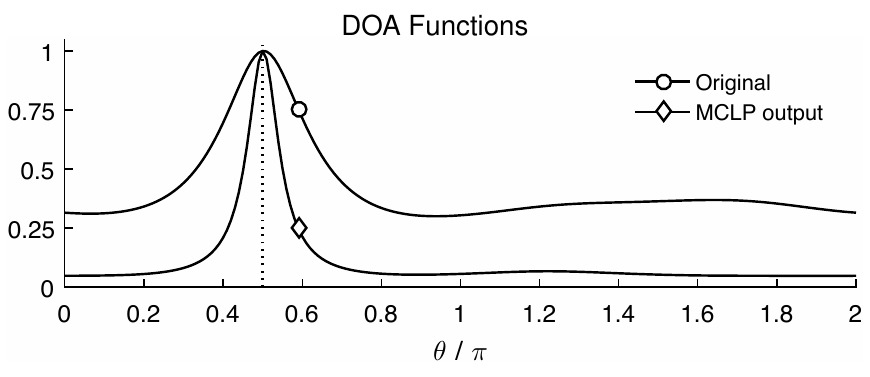}
\caption{Experiment~\ref{exp:DOA} : The DOA functions $D(k,\theta)$ (see \eqref{eqn:DOAfunct}) averaged over time (i.e., $s$) for the reverberant observations and MCLP outputs. We also apply a normalization to each DOA function so that the peak is equal to unity.\label{fig:DOA}}
\end{figure}
\end{experiment}

\begin{experiment}[Source Separation]\label{exp:ss}
In this experiment, we demonstrate the improvement gained by employing MCLP as a pre-processing step in a blind separation setup. We use the same microphone array as in Exp.~\ref{exp:DOA}. We use the source from Exp.~\ref{exp:DOA} as the first source. The second source is produced by a loudspeaker placed at $\theta = 1.75\,\pi$ radians (so that the sources are separated by $0.75\,\pi$  radians). The spectrogram of the observation at the first microphone is shown in Fig.~\ref{fig:SpectMixed}a. The output of MCLP when the first microphone is taken as the reference is shown in Fig.~\ref{fig:SpectMixed}b. 

\begin{figure}
\centering
\includegraphics[scale=1]{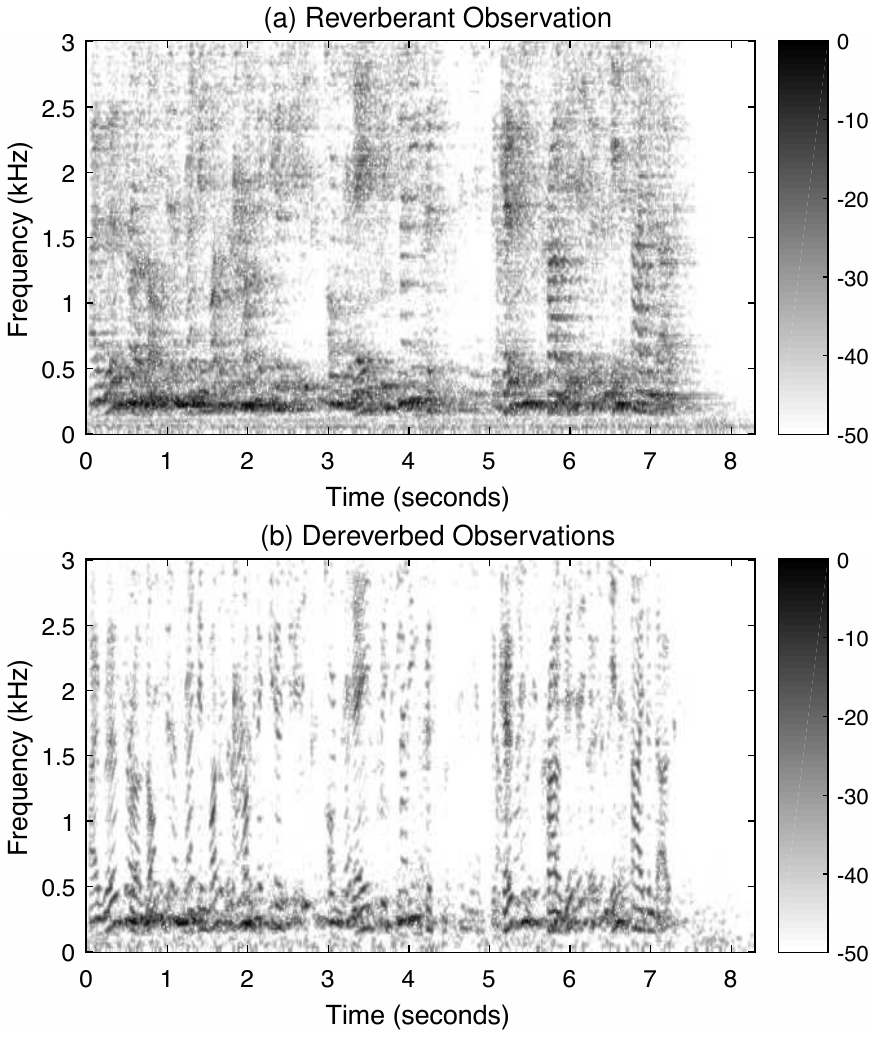}
\caption{ Experiment~\ref{exp:ss} : (a) Spectrogram of the two mixture signal observed at the first microphone. (b) MCLP applied to the observations, when the first microphone is taken as the reference. \label{fig:SpectMixed}}
\end{figure}

The average (over time) DOA functions obtained from the original recordings and MCLP outputs are shown in Fig.~\ref{fig:MixMUSIC}. We observe that MCLP has considerably sharpened the functions. Also, without MCLP, the peaks of the DOA functions are offset by some amount from the actual directions, which would lead to a biased estimate of the true DOA.

\begin{figure}
\centering
\includegraphics[scale=1]{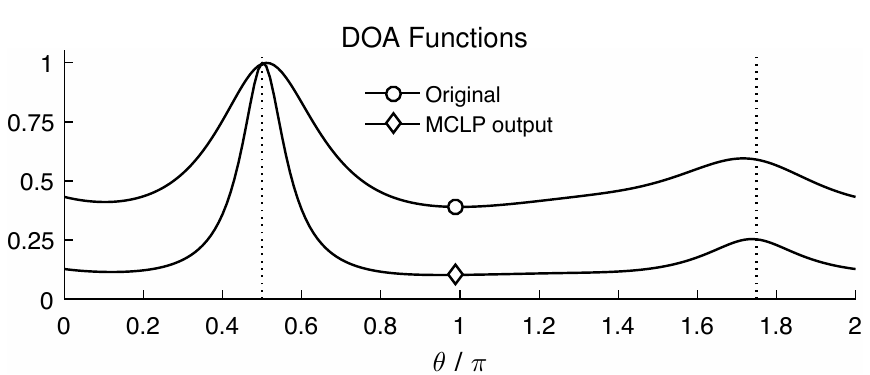}
\caption{Experiment~\ref{exp:ss} : The DOA functions $D(k,\theta)$ (see \eqref{eqn:DOAfunct}) averaged over time (i.e., $k$) for the reverberant observations and MCLP outputs. We also apply a normalization to each DOA function so that the peak is equal to unity.\label{fig:MixMUSIC}}
\end{figure}

The source separation results are shown in Fig.~\ref{fig:SpectSeparate}. Perceptually, without MCLP, especially for the second source, we can clearly hear the first source, despite the post-filter. Further, reverberation is not eliminated. These artifacts are eliminated using MCLP.

\begin{figure}
\centering
\includegraphics[scale=1]{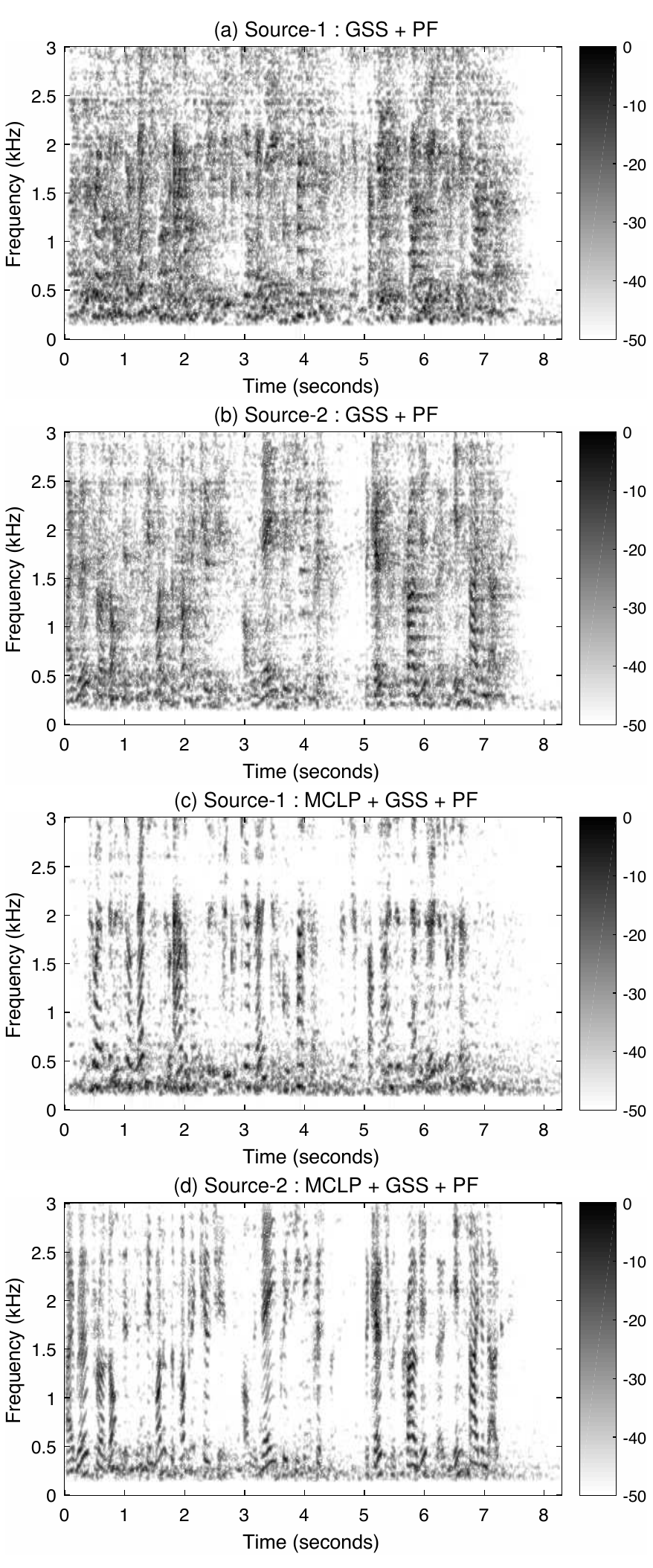}
\caption{Experiment~\ref{exp:ss} : Spectrograms of the sources separated with and without MCLP. (a,b) Source-1,2 obtained with GSS and a post-filter, without MCLP. (c,d) Source-1,2 obtained when MCLP is applied first.  \label{fig:SpectSeparate}}
\end{figure}

In order to quantify the results, we also report PESQ and SIR scores \cite{vin06p462}. PESQ gives an objective perceptual distance between an input speech signal and a reference. In order to obtain meaningful distances, we use the multiple channel recordings of source-1 and source-2, prior to mixing, to extract the reference signals. We take the $i\thh$ microphone as reference and apply MCLP to the multiple microphone recordings to obtain the references for each source. We then compute the PESQ scores for the methods described in this paper. We finally average over $i$ to obtain a reliable PESQ score. The resulting scores are tabulated in Table~\ref{table:PESQ1}. We remark that the inclusion of an MCLP step (denoted as LP in the table) improves the PESQ scores noticeably.

The SIR scores for separation with different combinations of the discussed methods are also provided in Table~\ref{table:PESQ1}. We observe that including an MCLP preprocessing step also improves significantly the separation quality measured by SIR.

\begin{table}
\renewcommand{\arraystretch}{1.3}
\centering
\caption{PESQ  and SIR scores for Exp.~\ref{exp:ss}. (LP : Multichannel Linear Prediction, G : Geometric Source Separation, P : Post filter)\label{table:PESQ1}}
\begin{tabular}{c}
PESQ
\end{tabular}\\
\begin{tabular}{c c c c c c c}
Reference & Obs. & LP &  G & G+P &  LP+G & LP+G+P\\
\hline
Source-1 & 1.07 & 1.04 & 1.05 & 1.08 & 1.10 & 1.20 \\
Source-2 & 1.04 & 1.08 & 1.05 & 1.14 & 1.22 & 1.49 \\
\hline
\end{tabular}
\vspace{0.5cm}

\begin{tabular}{c}
SIR
\end{tabular}\\
\begin{tabular}{c c c c}
G & G+P & LP+G & LP+G+P\\
\hline
5.0 & 13.0 & 7.8 & 19.5\\
\hline
\end{tabular}
\end{table}

\end{experiment}

\section{Conclusion}

In this paper, we advocated the use of MCLP as a pre-processing step for blind reverberant source separation. We provided theoretical justification for such a scheme and demonstrated numerically that it can significantly improve the performance of readily available separation algorithms.

The current algorithm for MCLP is computationally intensive. However, for a practical scenario with slowly moving sources, if we have approximate MCLP weights at a certain time, we think it would be possible to update the weights in a computationally efficient manner by warm-starting the optimization procedure. Therefore, the proposed scheme has potential for real-time applications as well.

Another issue involves the correlation between the weights obtained for different reference microphones. In the current formulation, MCLP is applied independently for each choice of the reference microphone. However, we would expect some correlation between these weights, especially if the microphones are close with respect to the room dimensions, which is usually the case for a microphone array. We plan to explore the implications of such correlation in future work.

\appendices
\section{Proof of Prop.~\ref{prop:modMINT}}\label{app:modMINT}
Without loss of generality, suppose $r = 1$ in the following.

Since $\mathcal{H}_2$, \ldots, $\mathcal{H}_N$ are co-prime, we can find causal polynomials $\tilde{\mathcal{U}}_2$, \ldots, $\tilde{\mathcal{U}}_N$ such that 
\begin{equation}
\sum_{k=2}^N \mathcal{\tilde{U}}_k(z) \,\mathcal{H}_k(z) = 1.
\end{equation}
Now set, $\mathcal{U}_1(z) = 1$ and 
\begin{equation}
\mathcal{U}_k(z) = \tilde{\mathcal{U}}_k (z) \, \frac{ \mathcal{H}_1(z) -  \mathcal{H}_1(\infty)}{z^{-1}}.
\end{equation}
Then for $\mathcal{G}$ as in \eqref{eqn:MCLPfilters}, and $d = 1$, we have 
\begin{align}
\sum_{k=1}^N \mathcal{G}_k(z) \,\mathcal{H}_k(z) & = 
\mathcal{H}_1(z)  - \bigl(\mathcal{H}_1(z) -  \mathcal{H}_1(\infty)\bigr) \\
&\quad \quad \quad\,\times\sum_{k=2}^N\,\mathcal{\tilde{U}}_k(z) \,\mathcal{H}_k(z)\\
&= \mathcal{H}_1(z) - \bigl(\mathcal{H}_1(z) -  \mathcal{H}_1(\infty)\bigr)\\
&= \mathcal{H}_1(\infty).
\end{align}

For the converse, observe that for causal $\mathcal{U}_k$ and $\mathcal{H}_k$, the polynomial
\begin{equation}
\mathcal{C}(z) = \sum_{k=1}^N z^{-d}\,\mathcal{U}_k(z)\,\mathcal{H}_k(z)
\end{equation}
does not have any constant terms, i.e., $C(\infty) = 0$. Observe also that $\mathcal{G}_1(\infty) = 1$.  Thus if \eqref{eqn:propMCLP} holds, then
\begin{align}
\left. \sum_{k=1}^N \mathcal{G}_k(z) \,\mathcal{H}_k(z) \right|_{z = \infty} &= \mathcal{G}_1(\infty)\,\mathcal{H}_1(\infty) - C(\infty)\\
& = \mathcal{H}_1(\infty).
\end{align}

\section{Canonical Form of a Polynomial Matrix}\label{app:Smith}
In this appendix, we briefly present, for the sake of completeness, the (Smith) canonical form of a polynomial matrix, along with some of its properties. For a more detailed exposition, we refer to the fine treatment in Sec.~VI.1--3 of Gantmacher's book \cite{Gantmacher1}.

We refer to a quantity or matrix as non-zero constant if it is a polynomial of order zero.

Let us start by introducing three \emph{elementary row operations} on a polynomial matrix $\bo{A}(z)$ :
\begin{enumerate}
\item Multiplication of a row with a non-zero constant.
\item Addition, to the $j\thh$ row, the product of the $i\thh$ row and a polynomial $p(z)$.
\item Interchange of two rows.
\end{enumerate}
These operations can be performed by left multiplication with a square matrix whose determinant is a non-zero constant. We refer to matrices that can be used to perform elementary operations as \emph{elementary matrices} in this appendix.  We also remark that similar operations can be defined on the columns instead of rows. Elementary column operations can be performed by right-multiplication with elementary matrices.

We say that two polynomial matrices $\bo{A}$ and $\bo{A}'$ are equivalent if $\bo{A}$ can be transformed to $\bo{A}'$ by elementary row/column operations.

Let us now introduce the (Smith) canonical form.
\begin{defn}\label{def:smith}
An  $N\times K$ (with $K \leq N$) diagonal polynomial matrix $\bo{A}(z)$ of rank $r$ with diagonal $a_i(z)$ for $i\in \{1,\ldots,K\}$ is said to be in (Smith) canonical form if 
\begin{enumerate}
\item $a_i(z) \neq 0$ if $i\leq r$, $a_i(z) = 0$ if $i>r$,
\item $a_i(0) = 1$ for $i\leq r$,
\item $a_i(z)$ divides $a_{i+1}(z)$ for $i<r$.
\end{enumerate}
\end{defn}

A fundamental result is that we can transform any polynomial matrix to one in canonical form via elementary row/column operations. 
\begin{prop}
Each $N \times K$ polynomial matrix is equivalent to a unique matrix in canonical form.
\end{prop}

In general, the canonical form is reached by performing both row and column operations. However, we are interested in row operations only, for the development in this paper. We have the following result.
\begin{prop}\label{prop:smith}
Suppose an $N\times K$ matrix $A$ has rank $K$ and the diagonal of its canonical form is constant. Then, $A$ can be reduced to its canonical form via row operations.
\end{prop}
This result is a corollary of the discussion in \cite{Gantmacher1} but is not stated explicitly. We provide a proof for the sake of completeness. We need a few intermediate results for the proof. All of the non-obvious claims stated before the proof of Prop.~\ref{prop:smith}  can be found in \cite{Gantmacher1}.

Given an $N\times K$ polynomial matrix $A(z)$, let $D_r(z)$ denote the greatest common divisor of all minors of $A$ of order $r$, normalized such that the leading coefficient is unity. Notice that $D_r$ is divisible by $D_{r-1}$. Therefore, for $r>1$, the ratios $i_r = D_r / D_{r-1}$ are polynomials. We also set $i_1 = D_1$. It can be shown that $D_r$'s are invariant under elementary row/column operations. This property transfers to $i_r$ as a consequence of its definition. 

Suppose now that the canonical form of $A(z)$ consists of a diagonal matrix with elements $a_r(z)$ as in Defn.~\ref{def:smith}. It turns out that $a_r = i_r$, giving us an alternative definition of $a_r$.

We need a final result before proving Prop.~\ref{prop:smith}, which can be shown by a simple elimination argument.
\begin{prop}\label{prop:U}
Using elementary row operations, any polynomial matrix $B(z)$ can be reduced to an upper triangular matrix $U(z)$ such that, at each column of $U(z)$, if the polynomial on the diagonal is non-zero, then the degree of any off-diagonal polynomial is strictly smaller than the degree of the polynomial on the diagonal.
\end{prop}

We are now ready for the proof of Prop.~\ref{prop:smith}.
\begin{proof}[Proof of Prop.~\ref{prop:smith}]
Suppose the canonical form of $A$ is $B$. Observe that the only non-zero minor of $B$ of order $K$ is given by the product of its diagonals, which is unity. Therefore, the invariant polynomial $D_K$ of $A$ is a constant.

Suppose we reduce $A$ to an upper triangular form $U$ as in Prop.~\ref{prop:U} using elementary row operations only. Observe that the only non-zero minor of $U$ of order $K$ is given by the product of its diagonals. Let us denote this non-zero minor by $U_K(z)$. But since this is the only non-zero minor of order $K$, and elementary matrices have constant determinants, it follows that we must have $U_K(z) = c\, D_K$ for some constant $c$. Thus, $U_K(z)$ is a constant. Therefore, $U$ is a constant diagonal matrix with a non-zero diagonal. Each constant on the diagonal can be made unity by an elementary row operation, yielding the canonical form $B$.
\end{proof}

In a final proposition, we show the equivalence stated in Remark~\ref{rem:equivalence}.
\begin{prop}\label{prop:equivalence}
Suppose $A$  is an $N\times K$  polynomial matrix with $K\leq N$.
The greatest common divisor of all the $K$-minors of $A$  is a non-zero constant if and only if the Smith canonical form of $A$ has a  non-zero constant diagonal.
\begin{proof}
As in the discussion above, let $D_r(z)$ denote the greatest common divisor of all minors of $A$ of order $r$, where the leading coefficient is unity, $i_r = D_r / D_{r-1}$ for $r>1$, with $i_1 = D_1$. Recall that $i_r$'s  form the diagonal in the Smith canonical form of $A(z)$.

Suppose now that $D_K(z) = 1$. This implies that $i_r = 1$ for all $r \leq K$. Therefore the diagonal of the canonical form of $A$  is the identity matrix.

For the converse, suppose the canonical form of $A$ has a constant diagonal. This means that $i_r(z)$'s are constant for $r \leq K$. But this implies that $D_r(z) = 1$  for $r\leq K$.
\end{proof}
\end{prop}
\section{Proof of Prop.~\ref{prop:MCLP2}}\label{app:MCLP2}
Without loss of generality, suppose $r = 1$ in the following.

By Remark~\ref{rem:equivalence}, which is proved in Prop.~\ref{prop:equivalence}, we first note that the gcd of the $K$-minors of  $\bar{\bo{H}}_1$ is a non-zero constant if and only if the canonical form of $\bar{\bo{H}}_1$ has a non-zero constant diagonal.
Thus, by Prop.~\ref{prop:smith}, we can find a $K \times (N-1)$ polynomial matrix $\bo{V}$ such that
\begin{equation}
\bo{V}\,\bar{\bo{H}}_1 = I.
\end{equation}
For $k = 1,\ldots, K$, let us define the polynomials
\begin{equation}
\mathcal{P}_k = \frac{\mathcal{H}_{1,k}(z) - \mathcal{H}_{1,k}(\infty)}{z^{-1}}
\end{equation}
and set
\begin{equation}
\begin{bmatrix}
\mathcal{U}_2 & \cdots & \mathcal{U}_N
\end{bmatrix} = \begin{bmatrix}
\mathcal{P}_1 & \cdots & \mathcal{P}_K
\end{bmatrix}\,\bo{V}.
\end{equation}
Also, set $\mathcal{U}_1(z) = 0$.
Then for $\mathcal{G}$ defined as in \eqref{eqn:MCLPfilters}, we have 
\begin{multline}
\begin{split}
\begin{bmatrix}
\mathcal{G}_1 & \cdots & \mathcal{G}_N
\end{bmatrix} \,
\bo{H} = \begin{bmatrix}
\mathcal{H}_{1,1}(z) & \cdots & \mathcal{H}_{1,K}(z)
\end{bmatrix} \\  \quad - z^{-1}\,\begin{bmatrix}
\mathcal{P}_1(z) & \cdots & \mathcal{P}_K(z)
\end{bmatrix} \\
= \begin{bmatrix}
\mathcal{H}_{1,1}(\infty) & \cdots & \mathcal{H}_{1,K}(\infty)
\end{bmatrix}
\end{split}
\end{multline}

For the converse, observe that, for causal $\mathcal{U}_k$ and $\mathcal{H}_{i,k}$,  the polynomial
\begin{equation}
\mathcal{C}_k(z) = \sum_{i=1}^N z^{-d}\,\mathcal{U}_i(z)\,\mathcal{H}_{i,k}(z)
\end{equation}
does not have any constant terms, regardless of $k$, i.e., $C_k(\infty) = 0$ for all $k$. Observe also that $\mathcal{G}_1(\infty) = 1$.  Thus if \eqref{eqn:propMCLP2} holds, then,
\begin{align}
\left. \sum_{i=1}^N \mathcal{G}_i(z) \,\mathcal{H}_{i,k}(z) \right|_{z = \infty} &= \mathcal{G}_1(\infty)\,\mathcal{H}_{1,k}(\infty) - C_k(\infty)\\
& = \mathcal{H}_{1,k}(\infty).
\end{align}

\section{Proof of Prop.~\ref{prop:main}}\label{app:main}

We remark that \eqref{eqn:propMCLP3} follows by expressing \eqref{eqn:propMCLP2} in the time-domain.

Let us consider \eqref{eqn:propMCLPDOA}. If \eqref{eqn:propMCLP2} holds, and $\mathcal{G}_i$'s satisfy \eqref{eqn:MCLPfilters}, then by \eqref{eqn:propcr}, we have
\begin{equation}
\sum_{i=1}^N \, \bigl( g_i \ast h_{i,k} \bigr)(s) = h_{r,k}(0)\,\delta(s).
\end{equation}
By Assumption~\ref{assm:main} and \eqref{eqn:prophi}, it follows that 
\begin{equation}
\sum_{i=1}^N \, \bigl( g_i \ast h_{i,k} \bigr)(s) = c'\,\exp(-j\,\omega\,\tau_{r,k})\,\delta(s),
\end{equation}
where $c'$  is independent of $r$  and $k$. Using \eqref{eqn:observation}, and the properties of convolution, we finally obtain
\begin{align}
\begin{split}
\sum_{i=1}^N\,g_i (s) \ast y_i(s) &= \sum_{i=1}^N\,g_i (s) \ast \Biggl( \sum_{k=1}^K\,x_k(s) \ast h_{i,k}(s) \Biggr)\\
& = \sum_{k=1}^K\,x_k(s) \ast \sum_{i=1}^N\,g_i (s) \ast h_{i,k}(s)\\
&= \sum_{k=1}^K\,x_k(s) \ast \Bigl( c'\,\exp(-j\,\omega\,\tau_{r,k})\,\delta(s) \Bigr)\\
&= c'\,\sum_{k=1}^K\,x_k(s) \,\exp(-j\,\omega\,\tau_{r,k}).
\end{split}
\end{align}

\section{Motivation for the Post-Filter in Sec.~\ref{sec:PF}}\label{app:PF}
In order to justify the variance expression in \eqref{eqn:sigma}, we will make use of Prop.~\ref{prop:PF} below. For simplicity, this proposition uses a seemingly independent notation. We will identify the elements in the proposition with our separation scenario right after the proposition. Recall that a complex valued Gaussian random variable $Z$ is said to be circular \cite{pic94p473} and have variance $2\,\sigma^2$  if its real and imaginary parts are independent and Gaussian random variables with variance $\sigma^2$. \begin{prop}\label{prop:PF}
Suppose $\gamma$ is an unknown constant and we observe $Y_n = c_n \gamma + u_n\,Z_n$, for $n = 1,\ldots, N$, where $|c_n|=1$, $u_n \in \mathbb{C}$ and $Z_n$'s are circularly symmetric, independent, unit variance, zero-mean Gaussian random variables. 

The uniformly minimum variance unbiased (UMVU) estimator for $\gamma$ is of the form 
\begin{equation}\label{eqn:UMVU}
\hat{\gamma} = \sum_{n=1}^N w_n \,Y_n
\end{equation}
where
\begin{equation}\label{eqn:wi}
w_n = c_n^*\,\beta\,|u_n|^{-2},
\end{equation}
with
\begin{equation}\label{eqn:beta}
\beta = \left( \sum_{n=1}^N |u_n|^{-2} \right)^{-1}.
\end{equation}
We also have that,
\begin{equation}
\frac{1}{N-1}\sum_{n=1}^N |w_n|\,\Bigl|  Y_n - c_n \hat{\gamma} \Bigr|^2
\end{equation}
is an unbiased estimator for the variance of $\hat{\gamma}$.
\begin{proof}[Proof (Sketch)]
It can be shown, by taking into account the Gaussian pdf of $Y_n$, that a complete sufficient statistic for $\gamma$ is
\begin{equation}
U = \sum_{n=1}^N c_n^*\,|u_n|^{-2}\,Y_n.
\end{equation}
Observing that $\mathbb{E}(U) = \beta^{-1} \, \gamma$, it follows by the Rao-Blackwell theorem that the UMVUE is as given in \eqref{eqn:UMVU} for $w_n$ as in \eqref{eqn:wi} and $\beta$ as in \eqref{eqn:beta}. It can also be checked that the variance of $\hat{\gamma}$  is,
\begin{equation}
\var(\hat{\gamma}) = \sum_{n=1}^N\,|w_n|^2\,|u_n|^2 =  \beta^2\,\sum_{n=1}^N\,|u_n|^{-2} = \beta. 
\end{equation}

Let us now set $\nu = \sum_{n=1}^N |w_n|\, \Bigl|\,  \bigl( Y_n - c_n \hat{\gamma} \bigr) \Bigr|^2$.
Then,
\begin{align}
\begin{split}
&\mathbb{E}(\nu) =   \sum_{n=1}^N  |w_n|\, \mathbb{E} \Bigl(  \bigl| Y_n - c_n \hat{\gamma} \bigr|^2 \Bigr)\\
& =  \sum_{n=1}^N  |w_n|\,\left( \bigl(1 - |w_n|\bigr)^2\,|u_n|^2 + \sum_{k \neq n} |w_k|^2\,|u_k|^2\right) \\
&= \sum_{n=1}^N |w_n|\,\left( \bigl(1 - |w_n|\bigr)^2\,|u_n|^2 + \beta\sum_{k \neq n} |w_k|\right) \\
&= \sum_{n=1}^N  |w_n|\,\left( \bigl(1 - |w_n|\bigr)^2\,\frac{\beta}{|w_n|} + \beta(1 - |w_n|) \right) \\
&= \sum_{n=1}^N  \beta (1-|w_n|) \\
&= \beta  (N-1).
\end{split}
\end{align}
\end{proof}
\end{prop}

In order to relate this proposition with the discussion in Section~\ref{sec:PF}, we make the following identifications : ${Y_n \to \hat{Y}_n(s,\omega)}$; ${c_n \to a_{\theta,n}}$; ${w_n \to W_{i,n}}$; ${\hat{\gamma} \to \hat{S}_i(s,\omega)}$. 

For the $i\thh$ source, if we regard all of the remaining sources as noise terms with a Gaussian distribution, then the optimal estimate of the $i\thh$ source is obtained by a linear combination of the observations. We assume that GSS approximately finds these optimal weights. In that case, the expression in \eqref{eqn:sigma} gives an estimate of the variance of these unwanted components (or, `noise') that remain after GSS.

\end{document}